\documentclass[11pt]{article}

\usepackage[usenames,dvipsnames,svgnames,table]{xcolor}
\usepackage{amssymb}
\usepackage{amsmath}
\usepackage{algorithm}
\usepackage{algpseudocode,enumerate}
\usepackage{alltt,fullpage,graphics,color,epsfig,amsthm, amssymb}
\usepackage{boxedminipage}
\usepackage[colorlinks=true,
            linkcolor=NavyBlue,
            urlcolor=blue,
            citecolor=OliveGreen,pagebackref]{hyperref}
\usepackage{tikz}
\usetikzlibrary{decorations.pathreplacing}

\usepackage[scaled]{helvet} 
\usepackage{courier} 
\normalfont
\usepackage[T1]{fontenc}
\usepackage{thmtools}
\newcommand{\old}[1]{}
\def\cR{{\cal R}}

\usepackage{amsthm}

\declaretheorem[numberwithin=section]{theorem}
\declaretheorem[sibling=theorem]{lemma}

\declaretheorem[sibling=theorem]{claim}

\declaretheorem[sibling=theorem]{corollary}

\declaretheorem[sibling=theorem]{remark}

\declaretheorem[sibling=theorem]{definition}

\usepackage{multicol}

\author{Anna R. Karlin\\ University of Washington\\ karlin@cs.washington.edu
	\and Shayan Oveis Gharan\\ University of Washington\\ shayan@cs.washington.edu 
	\and Robbie Weber\\ University of Washington\\ rtweber2@cs.washington.edu }
	
\begin{document}
\title{A Simply Exponential Upper Bound \\on the Maximum Number of Stable Matchings}
\maketitle
\begin{abstract}
Stable matching is a classical combinatorial problem that has been the subject of intense theoretical and empirical study since its introduction in 1962 in a seminal paper by Gale and Shapley~\cite{gale1962college}.
In this paper, we provide a new upper bound on $f(n)$, the maximum number of stable matchings that a stable matching instance with $n$ men and $n$ women can have.  It has been a long-standing open problem to understand the asymptotic behavior of $f(n)$ as $n\to\infty$, first posed by Donald Knuth in the 1970s \cite{knuth1976mariages}. Until now the best lower bound was approximately  $2.28^n$, and the best upper bound was $2^{n\log n- O(n)}$.
In this paper, we show that for all $n$, $f(n) \leq c^n$ for some universal constant $c$. This matches the lower bound up to the base of the exponent. 
Our proof is based on a reduction to counting the number of downsets of a family of posets that we call ``mixing''. The latter might be of independent interest.
\end{abstract}

\newpage
\section{Introduction} \label{sec:intro}
Stable matching is a classical combinatorial problem that has been the subject of intense theoretical and empirical study since its introduction in a seminal paper by Gale and Shapley in 1962~\cite{gale1962college}. Variants of the algorithm introduced in~\cite{gale1962college} are widely used in practice, e.g. to match medical residents to hospitals. Stable matching is even the focus of the 2012 Nobel Prize in Economics~\cite{nobel}. 

A stable matching instance with $n$ men and $n$ women is defined by a set of preference lists, one per person. Person $i$'s preference list gives a ranking over members of the opposite sex.  The {\em Stable Matching Problem} is to find a matching (i.e., a bijection) between the men and the women that
is {\em stable}, that is, has no {\em blocking pairs}. A man $m$ and a woman $w$ form a blocking pair in a matching if they are not matched to each other, but both prefer the other to their partner in the matching. 
Gale and Shapley~\cite{gale1962college} showed that a stable matching always exists and gave an efficient algorithm to find one.\footnote{Stable matching algorithms were actually developed and used as early as 1951 to match interns to hospitals~\cite{stalnaker1953matching}.}
Since at least one stable matching always exists, a natural question is to determine the maximum number of stable matchings an instance of a given size can have. This problem was posed in the 1970s in a monograph by Knuth \cite{knuth1976mariages}, and was the first of Gusfield and Irving's twelve open problems in their 1989 textbook \cite{gusfield1989stable}. We denote the maximum number of stable matchings an instance with $n$ men and $n$ women can have by $f(n)$. 

Progress on determining the asymptotics of $f(n)$ has been somewhat slow.
The best lower bound is approximately  $2.28^n$, and the best upper bound prior to this paper was $2^{n\log n- O(n)}$. See the related work section for a detailed history.

In this paper, we present an improved upper bound.
\begin{theorem}\label{thm:main}
	There is a universal constant $c$ such that $f(n)$, the number of stable matchings in an instance with $n$ men and $n$ women, is at most $c^n$. 
\end{theorem}

To prove this theorem, we use a result due to  Irving and Leather~\cite{irving1986complexity} that shows that there is a bijection between the stable matchings of an instance $\mathcal I$ and the downsets\footnote{See \autoref{sec:argumentOutline} for definitions of all the relevant terminology.} of a particular partially-ordered set (poset) associated with $\mathcal I$ known as the {\em rotation poset}. We show that the rotation poset associated with a stable matching instance has a particular property that we call {\em $n$-mixing}, and that any poset with this property has at most $c^n$ downsets. All the steps in our proof are elementary.

The bound extends trivially to stable roommates instances. In the stable roommates problem, a set of $n$ agents rank the other $n-1$ agents in the set. The agents are paired off into roommate pairs, which are stable if no two agents would like to leave their partners and be matched to each other.  
A construction of Dean and Munshi \cite{dean2010faster}, demonstrates that a stable roommates instance with $n$ agents can be converted into a stable matching instance with $n$ men and $n$ women, such that the stable roommate assignments correspond to a subset of the stable matchings in the new instance. Using this construction, we can apply our upper bound to Stable Roommates. 
\begin{theorem}
	There is a universal constant $c$, such that the number of stable assignments in a stable roommate instance with $n$ agents is at most $c^n$.
\end{theorem}  

\subsection{Related Work}

\paragraph{Lower bounds:}
It is trivial to provide instances with $2^{n/2}$ stable matchings by combining disjoint instances of size 2. Irving and Leather constructed a family of instances \cite{irving1986complexity} which has since been shown by Knuth\footnote{Personal communication, as described in \cite{gusfield1989stable}.} to contain at least $\Omega(2.28^n)$ matchings. Irving and Leather's family only has instances for $n$ which is a power of $2$. Benjamin, Converse, and Krieger also provided a lower bound on $f(n)$ by creating a family of instances with $\Omega(2^n\sqrt{n})$ matchings \cite{benjamin1995marry}. While this is fewer matchings than the instances in \cite{irving1986complexity}, Benjamin et al.'s family has instances for every even $n$, not just powers of $2$. In 2002, Thurber extended Irving and Leather's lower bound to all values of $n$. For $n$ powers of $2$, Thurber's construction exactly coincides with Irving and Leather's. For all other $n$, the construction produces a lower bound of $2.28^n / c^{\log n}$ for some constant $c$ \cite{thurber2002concerning}. To date, this lower bound of $\Omega(2.28^n)$ is the best known. We refer the reader to Manlove's textbook for a more thorough description of the history of these lower bounds \cite{manlove2013algorithmics}. 

\paragraph{Upper bounds:} Trivially, there are at most $n!$ stable matchings (as there are at most $n!$ bijections between the men and women). The first progress on upper bounds that we are aware of was made by Stathoupolos in his 2011 Master's thesis \cite{stathopoulos2011variants}, where he proves that the number of stable matchings is at most $O(n! / c^n)$ for some constant $c$. A more recent paper of Drgas-Burchardt and {\'S}witalski shows a weaker upper bound of approximately $\frac{3}{4} n!$ \cite{drgas2013number}. All previous upper bounds have the form $2^{n\log n - O(n)}$. 

\paragraph{Restricted preferences:}
The number of possible stable matchings has also been studied under various models restricting or randomizing the allowable preference lists. If all preference lists are equally likely and selected independently for each agent, Pittel shows that the expected number of stable matchings is $O(n \log n)$ \cite{pittel1989average}. Applying Markov's Inequality shows that the number of stable matchings is polynomial in $n$ with probability $1 - o(1)$. Therefore, the lower bound instances described above are a vanishingly small fraction of all instances. Work of Hoffman, Levy, and Mossel (described in Levy's PhD thesis \cite{levy2017novel}) shows that under a Mallows model~\cite{mallows1957non}, 
where preference lists are selected with probability proportional to the number of inversions in the list, the number of stable matchings is $C^n$ with high probability (where the constant $C$ depends on the exact parameters of the model). 

The number of attainable partners\footnote{ Woman $w$ is an attainable partner of man $m$ if there is a stable matching in which they are matched to each other.} a person can have has also been the subject of much research.  Knuth, Motwani, and Pittel show that the number of attainable partners is $O(\log n)$ with high probability if the lists are uniformly random \cite{knuth1990stable}. Immorlica and Mahdian show that if agents on one side of the instance have random preference lists of length $k$ (and consider all other agents unacceptable) the expected number of agents with more than one attainable partner depends only on $k$ (and not on $n$) \cite{immorlica2005marriage}. Ashlagi, Kanoria, and Leshno show that if the number of men and women is unbalanced, with uniformly random lists, the fraction of agents with more than one attainable partner is $1 - o(1)$ with high probability~\cite{ashlagi2017unbalanced}. 

\paragraph{Counting:}
A natural computational problem is to count the number of stable matchings in a given instance as efficiently as possible. Irving and Leather show that finding the exact number of matchings is $\#P$-complete \cite{irving1986complexity}, so finding an approximate count is a more realistic goal. Bhatnagar, Greenberg, and Randall consider instances where preference lists come from restricted models \cite{bhatnagar2008sampling}; for example, those in which the preference lists reflect linear combinations of $k$ ``attributes'' of the other set, or where every agent appears in a ``range'' of $k$ consecutive positions in every preference list. In both of these cases, they show that a natural Markov Chain Monte Carlo approach does not produce a good approximation (as the chain does not mix efficiently). 
As part of their proof, they show the number of stable matchings can still be as large as $c^n$ for some constant $c < 2.28$, even in these restricted cases. 

A formal hardness result was later shown by Dyer, Goldberg, Greenhill, and Jerrum \cite{dyer2004relative}. They show approximately counting the number of stable matchings is equivalent to approximately counting for a class of problems, canonically represented by $\#$BIS.\footnote{More specifically, they show an FPRAS for the number of stable matchings exists if and only if one exists for $\#$BIS, approximately counting the number of independent sets in a bipartite graph. Goldberg and Jerrum conjecture that no such FPRAS exists \cite{goldberg2012approximating}. See \cite{chebolu2012complexity} for formal definitions.} This hardness result was strengthened by Chebolu, Goldberg, and Martin \cite{chebolu2012complexity} to hold even if the instances come from some of the restricted classes of \cite{bhatnagar2008sampling}.

The heart of all of these results is the rotation poset (originally developed in \cite{irving1986complexity}), which we use and describe in \autoref{sec:rotationGraph}.

\paragraph{Stable matching in general:} See the books by Roth and Sotomayor~\cite{roth1992two}, Gusfield and Irving~\cite{gusfield1989stable}, Manlove~\cite{manlove2013algorithmics}, and Knuth~\cite{knuth1997stable} for more about the topic of stable matching. For many examples of stable matching in the real world, see~\cite{roth2015}.

\section{Preliminaries and main technical theorem} \label{sec:argumentOutline}

In this section, we review standard terminology regarding partially ordered sets, describe the key property of a poset we will use, and state our main technical theorem 
(\autoref{thm:mixingDownsets}). 

\begin{definition}[Poset]
	A partially ordered set (or poset) $(V, \prec)$, is defined by a 
	set $V$ and a binary relation, $\prec$, on $V$ satisfying: 
		\begin{itemize}
			\item {\bf Antisymmetry:} For all distinct $u,v \in V$ if $u \prec v$ then $v \not\prec u$,  and
			\item {\bf Transitivity:} for all $u,v,w \in V$ if $u \prec v$ and $v \prec w$ then $u \prec w$.
		\end{itemize} 
\end{definition}

\noindent	Two elements $u,v \in V$ are {\em comparable} if $u \prec v$ or $v \prec u$. They are {\em incomparable} otherwise.  

\noindent	If $u \prec v$, we say $u$ is {\em dominated} by $v$ and $v$ {\em dominates} $u$. 
\begin{definition}[Chain]
	A set $S$ of elements is called a {\em chain} if each pair of elements in $S$ is comparable. 
	In other words, for $\ell>0$, a chain of {\em length} $\ell$ is a sequence of elements $v_1 \prec v_2 \prec v_3 \prec \dots \prec v_\ell$.
\end{definition}
\noindent A set of elements is called an {\em antichain} if they are pairwise incomparable. 
\begin{definition}[Downset]
A {\em downset} of a partial order is an antichain and all elements dominated by some element of that antichain.
\end{definition}
\noindent Observe that a downset is  closed under $\prec$. That is, for any downset $S$, if $v \in S$ and $u \prec v$ then $u \in S$. 

The following is the key property of the posets associated with stable matching instances that we will use in the proof.
\begin{definition}[$n$-mixing]\label{def:nmixing}
	A poset $(V, \prec)$ is $n$-mixing if there exist $n$ chains $C_1, \ldots, C_n$  (not necessarily disjoint) such that
	\begin{enumerate}[i)]
		\item Every element of $V$ belongs to at least one chain, i.e., $\cup_{i=1}^n C_i = V$,
		\item  For any  $U \subseteq V$, there are at least $2\sqrt{|U|}$ chains each containing an element of $U$.
	\end{enumerate}
\end{definition}

\vspace{0.1in}
Observe that if a poset is 
formed by $n$ disjoint chains, each of length $\ell$, it has about $\ell^n$ downsets, and $\ell$ could be arbitrarily bigger than $n$. But such a  poset is not mixing, since taking $U$ to be the set of elements on one of the chains violates property ii) of mixing. For an example of a mixing poset, see \autoref{fig:mixing}. We can now state our main technical theorem.

\vspace{0.1in}
\begin{figure}
\begin{center}
\begin{tikzpicture}[rotate=45,scale=0.8]
	\foreach \i in {0, 1, 2, 3, 4,  5, 6}{
		\foreach \j in {0, 1, 2, 3, 4, 5, 6}{
			\ifthenelse  	{\equal{\j}{3} \OR \equal{\j}{4} \OR \equal{\i}{3} \OR \equal{\i}{4}	}{
				\node at (\i,\j) (\i_\j) {};
			}{
				\node [draw,circle,fill] at (\i,\j) (\i_\j) {}; 
			}
		}
		\ifthenelse {\equal{\i}{3} \OR \equal{\i}{4}}{}{
			\foreach \j/\k in {1/0, 2/1, 3/2, 5/4, 6/5}{
					\draw [color=blue,->,line width=1.1pt] (\i_\j) edge (\i_\k); 
			}
		}
	}
	\foreach \j in {0,1,2,5,6}{
		\foreach \i/\k in {1/0, 2/1, 3/2, 5/4, 6/5}{
			\draw [color=red,->,line width=1.1pt] (\i_\j) edge (\k_\j);
		}
	}
		\draw [decorate,decoration={brace,amplitude=10pt},xshift=-4pt,yshift=0pt]
(-0.3,-0.3) -- (-0.3,6.2) node [black,midway,xshift=-0.6cm] 
{$n$};
	\draw [dotted, line width=1.3pt]  (1,3) -- (1,4) (5.5,3) -- (5.5,4) (3,1)--(4,1) (3,5.5)--(4,5.5) ; 
		\draw [decorate,decoration={brace,amplitude=10pt,mirror},xshift=4pt,yshift=0pt]
(-0.5,-0.3) -- (6.2,-0.3) node [black,midway,yshift=-0.6cm] 
{$n$};
\end{tikzpicture}
\caption{A $2n$-mixing poset with respect to chains defined by the red and blue paths. A path in the graph from $v$ to $u$ indicates that $u \prec v$. That is, the poset is the transitive closure of the arrows shown. For any set $U$ of $k$ elements, there are at least $2\sqrt{k}$ chains that contain one of these elements. }
\label{fig:mixing}
\end{center}
\end{figure}
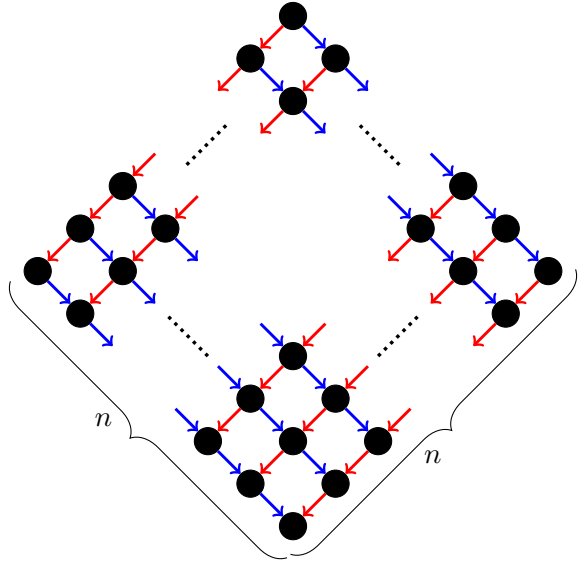

\begin{theorem}
\label{thm:mixingDownsets}
	There is a universal constant $c$, such that if a poset is $n$-mixing, then it has at most $c^ n$ downsets. 
\end{theorem}
 
Note that the $n$-mixing property immediately implies that the poset has at most $n^2$ elements; just let $U=V$ in the above definition. A poset with $n^2$ elements covered by $n$ chains can have at most $(n+1)^{n}$ downsets (this is achieved for $n$ equal length chains). So, the main contribution of the above theorem is to improve this trivial upper bound to $c^n$ for some constant $c$.

\autoref{thm:mixingDownsets} is the main technical contribution of this work. The proof is contained in \autoref{sec:MainProof}.
To complete the proof of \autoref{thm:main} we use the following theorem relating mixing posets to stable matchings.
\begin{restatable}{theorem}{rotationReduction}
\label{thm:rotationreduction}
For every stable matching instance $\mathcal I$ with a total of $n$ men and women, there exists an $n$-mixing poset $(\cR,\prec)$, called the rotation poset, such that the number of downsets of the poset is equal to the number of stable matchings of $\mathcal I$.	
\end{restatable}
 Note that \autoref{thm:mixingDownsets} and \autoref{thm:rotationreduction} immediately imply \autoref{thm:main}. 
We prove \autoref{thm:rotationreduction} in  \autoref{sec:rotationGraph}, by combining existing observations about the rotation poset.
\section{Proof of main technical theorem}\label{sec:MainProof} 
In this section we prove \autoref{thm:mixingDownsets}.
The proof proceeds by finding an element $v$ of the poset which dominates and is dominated by many elements. We then count downsets by
considering those downsets that contain
$v$ and those that do not. Since $v$ dominates and is dominated by many
elements, the size of each remaining instance is significantly smaller, yielding the bound.

Formally, we say an element is  {\em $\alpha$-critical} if it dominates $\alpha$ elements and is dominated by $\alpha$ elements.
The key lemma 
is that there is always an $\Omega((|V|/n)^{3/2})$-critical element.

\begin{lemma} \label{lma:criticalNode}
	Let $(V, \prec)$ be an $n$-mixing poset with respect to chains $C_1, \ldots, C_n$, and define $d = \frac{|V|}{n}$. For some universal constants $d_0>1$ and $c_0>0$,  there is an element $v \in V$ such that $v$ is $(c_0 d^{3/2})$-critical as long as $d \ge d_0$.
\end{lemma} 
\noindent We prove \autoref{lma:criticalNode} in \autoref{sec:criticalNode} via a counting argument. 

In the rest of this section we prove \autoref{thm:mixingDownsets} using \autoref{lma:criticalNode}. 
We bound the number of downsets by induction on $d$.

	Our base case is when $d = d_0$. 
	In this case, the number of downsets is maximized when the chains are all the same length, so we have an upper bound of $(d_0+1)^n$. 
	
 For larger $d$, first we identify a $\left(c_0 d^{3/2}\right)$-critical element $v$. The number of downsets containing $v$ is the number of downsets in
 the poset remaining after we delete $v$ and everything it dominates. Similarly, the number of downsets not containing $v$ is
 the number of downsets in the poset remaining after we delete $v$ and everything dominating $v$. In both cases, the resulting poset is still $n$-mixing, so we can induct. 
 We call such a step (choosing a critical element) an {\em iteration}.
 
It remains to bound the number of downsets that this process enumerates. 
By the $n$-mixing property, there are at most $n^2$ elements in the initial poset. 
We partition the iterations into phases, where in phase $i$ we reduce the size of the poset from $\frac{n^{2}}{2^i}$ elements to $\frac{n^2}{2^{i+1}}$ elements. By definition, in phase $i$, $d \geq \frac{n}{2^{i+1}}$. So, we can bound the number of iterations required in phase $i$ (call it $k_i$) by:
\[  c_0 \left( \frac{n}{2^{i+1}} \right)^{3/2} k_i >  \frac{n^2}{2^{i+1}}.\]
Rearranging, we see that it suffices to choose $k_i = 2^{(i+1)/2}\sqrt{n}/c_0 $. We continue until $d = d_0$. Summing across all phases, the number of choices to make is at most
\[ \frac{\sqrt{n}}{c_0}\sum \limits_{i=0}^{\log n} 2^{(i+1)/2} < 5\frac{n}{c_0}. \]
So, the algorithm enumerates at most $(d_0+1)^n$ downsets in the base case and it makes at most $5n/c_0$ choices during the inductive process. Thus, the number of downsets is at most $(d_0 + 1)^n 2^{5n/c_0} = c^n$ as required.  

\subsection{Proof of main technical lemma} \label{sec:criticalNode}

Finally, we prove \autoref{lma:criticalNode}, i.e. we show that any $n$-mixing poset $(V, \prec)$ contains a $c_0{d^{3/2}}$-critical element as long as $d \ge d_0$. (Recall that $d= |V|/n$.)

We make use of the standard graph representation of the partial order: In this graph there is a node for each element of the partial order, and  a directed edge from $v$ to $u$ if $u \prec v$. Of course,
this directed graph is acyclic. Henceforth, we refer only to this DAG rather than to the poset and partition the nodes of the DAG into {\em levels} as follows: Level $1$ nodes are those with no outgoing edges, i.e., sinks of the DAG, and level $i$ nodes are those whose longest path to a sink (i.e., a level $1$ node) has exactly $i$ nodes. Note that each level is an antichain.

Next, we create $n$ disjoint subchains $S_1, \ldots, S_n$. The subchain $S_i$ will be a subset of the nodes in $C_i$. We perform the assignment of nodes to subchains by processing up the DAG level by level.  Initialize every subchain $S_i$ to be empty. For each node $u$, consider the set of indices $I(u)  = \{j : u \in C_j\}$. Assign $u$ to the subchain for an index in $I(u)$ which currently has the fewest nodes among those chains, i.e. $\arg\min_{j \in I(u)} |S_j|$, breaking ties arbitrarily.   If $u$ is the $k^{th}$ node
assigned to a subchain $S_i$, then we say the {\em height} of $u$ is $k$.
By construction, the following properties hold: 

(a) The $S_i$'s are chains, since $S_i\subseteq C_i$. 

(b) The $S_i$'s  are disjoint. 

(c) If $u$ has height $h$, then $u$ dominates at least $h-1$ nodes
in each of the subchains $S_j$ such that $j\in I(u)$ (i.e., $\{j : u \in C_j\}$).

\begin{claim}
Let $\overline{D}$ be the set of nodes  of height at least $\lfloor d/2\rfloor$.
Each node  $u\in \overline{D}$ dominates at least  $c_0 \cdot d^{3/2}$ nodes,
for $d \ge d_0$, where $c_0$ and $d_0$ are universal constants.

\end{claim}

\begin{proof}
Suppose that $u\in S_i\cap \overline{D}$ is at height $\ell \ge \lfloor d/2\rfloor$ and let $D(u)$ be
the nodes in $S_i$ of height $\lceil\ell/2\rceil$ through $\ell$.
By the mixing property, these nodes lie on at least $2 \sqrt{\lfloor\ell/2\rfloor}$
chains $\mathcal C$. Moreover, by construction, on each
subchain $S_j \in \mathcal C$, at least $\lfloor\ell/2\rfloor - 1$ nodes are dominated
by some node in $D(u)$ and hence are dominated by $u$.
Therefore, $u$ dominates at least $(\lfloor\ell/2\rfloor -1)\cdot 2 \sqrt{\lfloor\ell/2\rfloor}  +  (\lfloor\ell/2\rfloor - 1) = \Omega(d^{3/2})$ nodes.
\end{proof}
Since the number of nodes in the DAG is 
$dn$, we conclude:\begin{corollary}
There is a set $\overline{D}$ of strictly more than $|V|/2$ nodes, that each dominate $\Omega(d^{3/2})$ nodes.
\end{corollary}
A symmetric argument in which subchains are built starting from the sources
of the DAG shows that there  is a set $\underline{D}$ of strictly more than $|V|/2$ nodes that are dominated by $\Omega(d^{3/2})$ nodes. Therefore, there is some node $v$ in the intersection of  $\overline{D}$ and $\underline{D}$.
This is the $c_0 d^{3/2}$-critical node we seek.

\begin{remark}
A crude analysis shows that $c_0 \ge 1/8$ when $d_0 > 25$.
\end{remark}

\section{Rotations and the rotation poset} \label{sec:rotationGraph}
In this section we present a key theorem of Irving and Leather~\cite{irving1986complexity} and show
how it can be used to prove \autoref{thm:rotationreduction}.

\rotationReduction* 

\noindent
We begin with the definitions needed to prove \autoref{thm:rotationreduction}.
\begin{definition}[Rotation]
Let $k \ge 2$.
A {\em rotation} $\rho$  is an ordered list of pairs $$\rho=((m_0, w_0), (m_1, w_1), \ldots, (m_{k-1}, w_{k-1}))$$ that are matched in some stable matching $M$ with the property that for every $i$ such that $0\le i \le k-1$, woman $w_{i+1}$ (where the subscript is taken mod $k$)
is the highest ranked woman on $m_i$'s preference list satisfying:\begin{enumerate}
		\item [i)] man $m_i$ prefers $w_i$  to $w_{i+1}$, and
		\item [ii)] woman $w_{i+1}$ prefers
$m_i$ to $m_{i+1}$.
	\end{enumerate}
	In this case, we say $\rho$ is {\em exposed} in $M$.
	\label{defn:rotation}
\end{definition}

We will sometimes abuse notation and think of a rotation as the set containing those pairs. 
Also, we will need the following facts about rotations later.
\begin{lemma}[{\cite[Lemma 4.7]{irving1986complexity}}] \label{fact:oneAppearance}
	A pair $(m,w)$ can appear in at most one rotation.
\end{lemma}

\begin{lemma}[{\cite[Lemma 2.5.1]{gusfield1989stable}}] \label{fact:noSkipping}
	If $\rho$ is a rotation with consecutive pairs $(m_i, w_i),$ and  $(m_{i+1}, w_{i+1})$, and $w$ is a woman between $w_i$ and $w_{i+1}$ in $m_i$'s preference list, then there is no stable matching containing the pair $(m_i, w)$. 
\end{lemma}

\begin{definition}[Elimination of a Rotation]
Let $\rho=((m_0, w_0), \ldots, (m_{k-1}, w_{k-1}))$ be a rotation exposed in stable matching $M$.  The rotation $\rho$ is {\em eliminated} from $M$ by
matching $m_i$ to $w_{(i+1)\bmod k}$, for all $0\leq i\leq k-1$, leaving all other pairs in $M$ unchanged, i.e., matching $M$ is replaced with matching $M'$, where
$$M'~:= M \backslash \rho ~\cup~\{(m_0, w_1), (m_1, w_2), \ldots, (m_{k-1}, w_0) \}.$$

Note that when we eliminate a rotation from $M$, the resulting matching $M'$ is stable.\footnote{Switching from $M$ to $M'$ makes all the women in $\rho$ happier and all the men in $\rho$ less happy. It is easy to check that this switch cannot create a blocking pair inside the set $\rho$. The only other possibility for a  blocking pair is a man in $\rho$ with a woman outside $\rho$. For $(m_i \in \rho,w \not \in \rho)$ to become a blocking pair, $m_i$ would have to prefer $w$ to $w_{i+1}$, but by the definition of rotation, $w_{i+1}$ was the first woman on $m$'s list who would prefer to be matched to him, so he cannot prefer $w$ to $w_{i+1}$. See also \cite[Lemma 2.5.2]{gusfield1989stable}.}
\end{definition}

\vspace{0.1in} 
Irving and Leather studied the following process: Fix a stable matching instance $\mathcal I$. Starting at the man-optimal matching\footnote{ A fascinating fact about stable matching is that there is a matching, known as the man-optimal matching, in which each man is matched with his favorite attainable partner. Recall that woman $w$ is attainable for man $m$ if there is some stable matching in which they are matched.}, choose a rotation in the current matching, and eliminate it. They show that for any stable matching $M$, there is a set of rotations, $R(M)$, one can eliminate (starting from the man-optimal matching) that will yield $M$. 

However, there is a partial order on the set of rotations -- some must be eliminated before others.\footnote{For example, a rotation containing a pair $(m, w)$ is not exposed (and thus cannot be eliminated) until that pair is matched, so a rotation with consecutive pairs $(m, w'), (m', w)$ must be eliminated first. The details of exactly when one rotation must be eliminated before another are not of direct use to us (we only require the rather coarse description in \autoref{lma:sharedAgentComparable}), so we do not describe them here. See \cite{irving1986complexity} or \cite{gusfield1989stable} for a full description of the poset.}
If $\rho$ must be eliminated before $\rho '$, we write $\rho \prec \rho'$. 
Let $\cR$ be the set of rotations for a stable matching instance, and let $\prec$ be that partial order on the rotations defined by elimination order. We call this poset the {\em rotation poset.} See \autoref{fig:examplePoset} for an example of a stable matching instance and the corresponding rotation poset.

For our purposes, the important result relating the rotation poset to stable matchings is the following.
\begin{theorem}[{\cite[Theorem 4.1]{irving1986complexity}}] \label{thm:downsetsToMatchings}
	For any stable matching instance, there is a one-to-one correspondence between downsets of the rotation poset and the set of stable matchings. In other words, the number of stable matchings is exactly equal to the number of downsets of $(\cR,\prec)$. 
\end{theorem}
Indeed, the downset corresponding to $M$ is exactly the set $R(M)$  discussed above.

Thus, to prove \autoref{thm:rotationreduction}, it remains to show that the rotation poset associated to any stable matching instance with a total of $n$ men and women is $n$-mixing. 
First we construct the chains $C_1,\dots,C_n$. We are going to have one chain for each agent (man or woman), where the corresponding chain contains all the rotations that include that agent.
Call these sets $C_1,\dots,C_n$.
To prove that $C_1,\dots,C_n$ is $n$-mixing, we first need to show that each $C_i$ is indeed a chain,  i.e., every pair of rotations where a specific agent appears are comparable and second we need to show that $C_1,\dots,C_n$ satisfy property (ii) of \autoref{def:nmixing}. 

\begin{claim} \label{lma:sharedAgentComparable}
If two rotations share an agent, then they are comparable.\footnote{This observation is not novel; for example, it is implicit in the discussion of \cite{gusfield1989stable}, but we have not seen the statement explicitly written down, so we prove it here.}
\end{claim}
\begin{proof}   
	First, suppose that the shared agent is a man. If it is a woman, we can just switch the designations of ``men'' and ``women'' and use the same proof on the ``reversed'' version of the rotation graph.
	Let $\rho_1,\rho_2$ be rotations sharing an agent $m$, and let $m$ be matched to $w_1$ in $\rho_1$ and $w_2$ in $\rho_2$, where $m$ prefers $w_1$ to $w_2$. 
	
	For the sake of contradiction assume that $\rho_1$ and $\rho_2$ are incomparable.
	We show that there exists a rotation $\rho$ which causes $m$ to skip over $w_1$. This would contradict \autoref{fact:noSkipping} as it implies that $(m,w_1)$ belongs to no stable matching. 
	  
	Suppose we start from the man optimal stable matching and eliminate all rotations dominated by $\rho_2$ 
	and let $M$ be the resulting stable matching.
	By the correspondence in \autoref{thm:downsetsToMatchings}, $(m,w_2)\in M$. Since $m$ prefers $w_1$ to $w_2$, there must be a rotation $\rho$ that we eliminated which caused $m$ to be matched to someone worse than $w_1$ for the first time. Since $\rho_2$ and $\rho_1$ are incomparable, $\rho_1\neq \rho$. Therefore, by  \autoref{fact:oneAppearance} $(m,w_1)\notin \rho$; so, $\rho$ caused  $m$ to skip over	 $w_1$. This is a contradiction.
\end{proof}

\begin{claim} \label{lma:manyAgents}
Every set of $k$ rotations contains at least $2\sqrt{k}$ agents.
\end{claim}
\begin{proof}
	We argue by contrapositive. Suppose we have a set of rotations involving fewer than $2\sqrt{k}$ agents. Every rotation contains a (man, woman) pair, who (by \autoref{fact:oneAppearance}) have not appeared together before. With fewer than $2\sqrt{k}$ agents, there are strictly less than $k$ (man, woman) pairs which can appear, and thus fewer than $k$ rotations in the set.  
\end{proof} 

\section{Conclusion} \label{sec:conclusion}
We have shown there is some constant $c$ such that $f(n) \leq c^n$. We have not made a significant effort to optimize the constants in our argument, favoring ease of exposition over the exact result. By making a few minor changes to the argument,  
we obtain $f(n) \leq 2^{17n}$ for sufficiently large $n$. A more careful argument could probably improve this constant somewhat, but this approach will not get a constant $c$ close to the (approximately) $2.28$ we would need to match the best known lower bound. Determining the precise asymptotic behavior of $f(n)$ remains an interesting open problem.

\section*{Acknowledgements} 

The first author gratefully acknowledges the support of NSF grant
CCF 1420381.
The second author gratefully acknowledges the support of NSF grant CCF-1552097 and ONR-YIP grant N00014-17-1-2429.

\begin{figure}

\begin{multicols}{2}

\includegraphics[width=.5\textwidth]{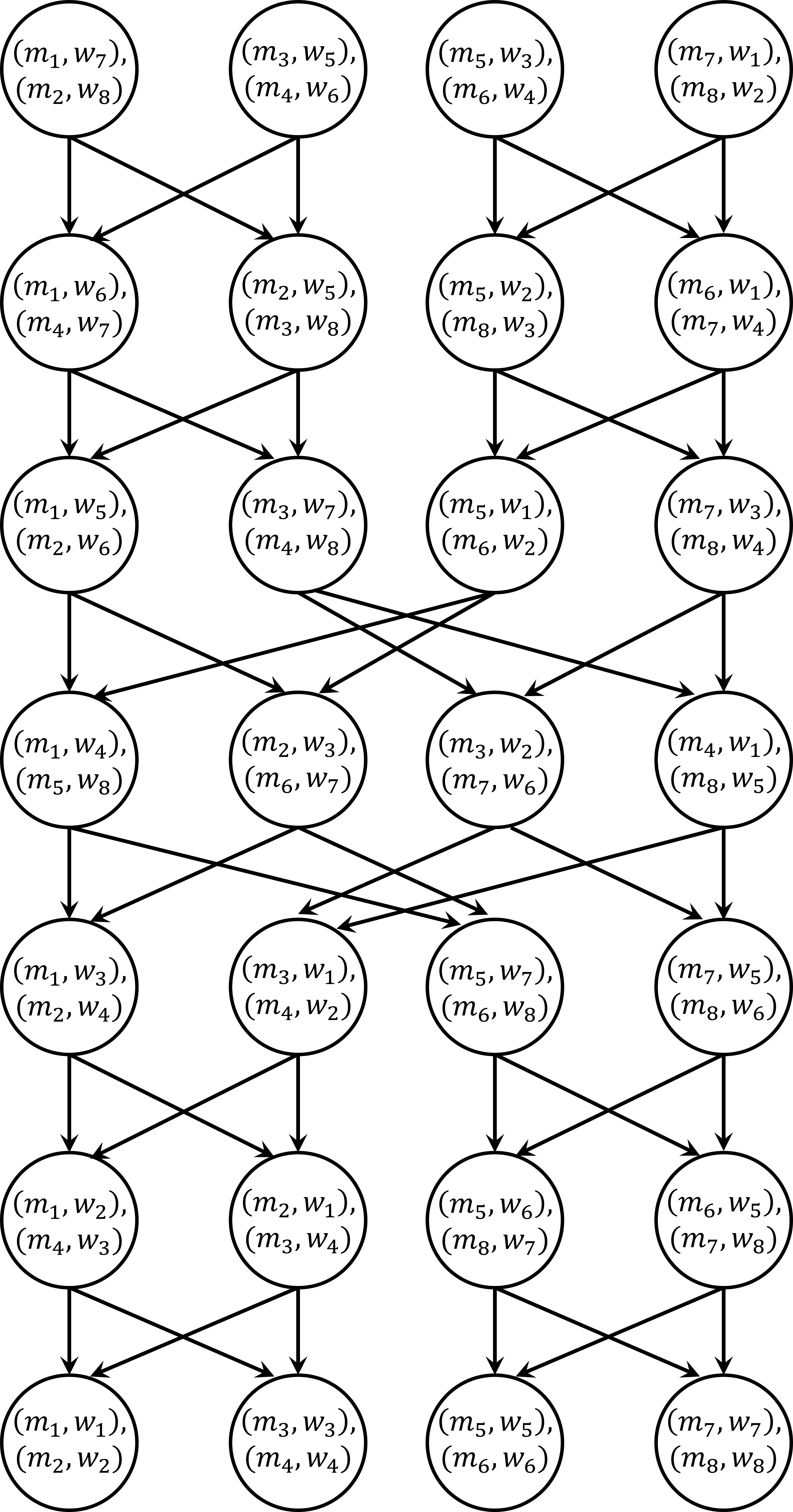}
\\
\vspace{0.5in}

Men's preferences:\\
\begin{tabular}{ | c || c | c | c | c || c | c | c | c |}
	\hline
	 
	$m_1$ & $w_1$ & $w_2$ & $w_3$ & $w_4$ & $w_5$ & $w_6$ & $w_7$ & $w_8$ \\ \hline
	$m_2$ & $w_2$ & $w_1$ & $w_4$ & $w_3$ & $w_6$ & $w_5$ & $w_8$ & $w_7$ \\ \hline
	$m_3$ & $w_3$ & $w_4$ & $w_1$ & $w_2$ & $w_7$ & $w_8$ & $w_5$ & $w_6$ \\ \hline
	$m_4$ & $w_4$ & $w_3$ & $w_2$ & $w_1$ & $w_8$ & $w_7$ & $w_6$ & $w_5$ \\ \hline \hline
	$m_5$ & $w_5$ & $w_6$ & $w_7$ & $w_8$ & $w_1$ & $w_2$ & $w_3$ & $w_4$ \\ \hline
	$m_6$ & $w_6$ & $w_5$ & $w_8$ & $w_7$ & $w_2$ & $w_1$ & $w_4$ & $w_3$ \\ \hline
	$m_7$ & $w_7$ & $w_8$ & $w_5$ & $w_6$ & $w_3$ & $w_4$ & $w_1$ & $w_2$ \\ \hline
	$m_8$ & $w_8$ & $w_7$ & $w_6$ & $w_5$ & $w_4$ & $w_3$ & $w_2$ & $w_1$ \\  \hline
\end{tabular}
	
	\vspace{1in}
Women's preferences:\\
\begin{tabular}{ | c || c | c | c | c || c | c | c | c |}	
\hline
	$w_1$ & $m_8$ & $m_7$ & $m_6$ & $m_5$ & $m_4$ & $m_3$ & $m_2$ & $m_1$ \\ \hline 
	$w_2$ & $m_7$ & $m_8$ & $m_5$ & $m_6$ & $m_3$ & $m_4$ & $m_1$ & $m_2$ \\ \hline
	$w_3$ & $m_6$ & $m_5$ & $m_8$ & $m_7$ & $m_2$ & $m_1$ & $m_4$ & $m_3$ \\ \hline
	$w_4$ & $m_5$ & $m_6$ & $m_7$ & $m_8$ & $m_1$ & $m_2$ & $m_3$ & $m_4$ \\ \hline \hline
	$w_5$ & $m_4$ & $m_3$ & $m_2$ & $m_1$ & $m_8$ & $m_7$ & $m_6$ & $m_5$ \\ \hline
	$w_6$ & $m_3$ & $m_4$ & $m_1$ & $m_2$ & $m_7$ & $m_8$ & $m_5$ & $m_6$ \\ \hline
	$w_7$ & $m_2$ & $m_1$ & $m_4$ & $m_3$ & $m_6$ & $m_5$ & $m_8$ & $m_7$ \\ \hline
	$w_8$ & $m_1$ & $m_2$ & $m_3$ & $m_4$ & $m_5$ & $m_6$ & $m_7$ & $m_8$ \\ \hline 
\end{tabular}

\end{multicols}

\caption{A size-$8$ instance of stable matching with its rotation poset. This is part of the family of instances described by Irving and Leather, which produces the $\Omega(2.28^n)$ lower bound.}
\label{fig:examplePoset}
\end{figure}

 \bibliographystyle{alpha}
\bibliography{matching_bib}

\end{document}